\documentclass{IEEEtran}

\usepackage{graphics} 
\usepackage{amsmath} 
\usepackage{amssymb}  
\usepackage{bm}
\usepackage{subfigure}
\usepackage{acronym}
\usepackage{hyperref}
\usepackage{paralist}
\usepackage{float}
\usepackage{color}
\usepackage{multicol}
\usepackage{tikz}

\makeatletter
\hypersetup{colorlinks=true}
\AtBeginDocument{\@ifpackageloaded{hyperref}
  {\def\@linkcolor{blue}
  \def\@anchorcolor{red}
  \def\@citecolor{red}
  \def\@filecolor{red}
  \def\@urlcolor{black}
  \def\@menucolor{red}
  \def\@pagecolor{red}
\begingroup
  \@makeother\`%
  \@makeother\=%
  \edef\x{%
    \edef\noexpand\x{%
      \endgroup
      \noexpand\toks@{%
        \catcode 96=\noexpand\the\catcode`\noexpand\`\relax
        \catcode 61=\noexpand\the\catcode`\noexpand\=\relax
      }%
    }%
    \noexpand\x
  }%
\x
\@makeother\`
\@makeother\=
}{}}
\makeatother

\usepackage{amsthm}

\newtheorem{Theorem}{Theorem}

\newtheorem{Lemma}{Lemma}

\newtheorem{Remark}{Remark}

\newtheorem{Assumption}{Assumption}

\newtheorem{Definition}{Definition}

\def\BibTeX{{\rm B\kern-.05em{\sc i\kern-.025em b}\kern-.08em
    T\kern-.1667em\lower.7ex\hbox{E}\kern-.125emX}}

\begin{document}
\title{Finite-Time Stability of Hybrid Systems: A Multiple Generalized Lyapunov Functions Approach}

\author{Kunal Garg, and Dimitra Panagou, \IEEEmembership{Member, IEEE}
\thanks{
The authors would like to acknowledge the support of the Air Force Office of Scientific Research under award number FA9550-17-1-0284.}
\thanks{The authors are with the Department of Aerospace Engineering, University of Michigan, Ann Arbor, MI, USA; \texttt{\{kgarg, dpanagou\}@umich.edu}.}
}

\maketitle

\begin{abstract}
This paper studies finite-time stability of a class of hybrid systems. We present sufficient conditions in terms of multiple generalized Lyapunov functions for the origin of the hybrid system to be finite-time stable. More specifically, we show that even if the value of the generalized Lyapunov functions increase between consecutive switches, finite-time stability can be guaranteed if the finite-time convergent mode is active long enough. In contrast to earlier work where the Lyapunov functions are required to be decreasing during the continuous flows and non-increasing at the discrete jumps, we allow the generalized Lyapunov functions to increase \emph{both} during the continuous flows and the discrete jumps. As thus, the derived stability results are less conservative compared to the related literature. Numerical example demonstrates the efficacy of the proposed methods. 
\end{abstract}


\section{Introduction}
Hybrid systems exhibit continuous state evolution and discrete state jumps, and therefore can capture the behavior of complex dynamical systems. The introductory paper \cite{antsaklis1998hybrid} provides an overview and merits of hybrid systems. Stability of switched and hybrid systems has been studied extensively; for a detailed presentation of the solution concepts and the notion of stability for hybrid systems, the interested readers are referred to \cite{lygeros2004lecture,goebel2012hybrid}.

The survey paper \cite{teel2014stability} studies Lyapunov stability (LS), Lagrange stability and asymptotic stability (AS) for stochastic hybrid systems (SHS), and provides Lyapunov conditions for stability in probability. The paper also presents open problems on converse results on the stability in probability of SHS. 
The review paper \cite{marty2015elements} discusses sufficient conditions for stability of three classes of hybrid systems, namely, (i) systems modeled on a time scale, (ii) impulsive hereditary systems, and (iii) weakly coupled systems defined on Banach spaces. In \cite{liu2016lyapunov}, the authors study hybrid systems exhibiting delay phenomena such as memory. They establish sufficient conditions for AS using Lyapunov-Razumikhin functions and Lyapunov-Krasovskii functionals. More recently, pointwise AS of hybrid systems is studied in \cite{goebel2016notions}, where the notion of set-valued Lyapunov functions is used to establish sufficient conditions for AS of a closed set. In \cite{teel2013lyapunov}, the authors impose an average dwell-time for the discrete jumps and devise Lyapunov-based sufficient conditions for exponential stability of closed sets; see \cite{goebel2016notions,teel2013lyapunov} for details on the notion of stability of closed sets. 

In contrast to AS, which pertains to convergence as time goes to infinity, finite-time stability (FTS)\footnote{With slight abuse of notation, we use FTS to denote the phrase "finite-time stability" or "finite-time stable", depending on the context.} is a concept that requires convergence of the solution in finite time \cite{bhat2000finite}. 
FTS of switched and hybrid systems has gained popularity in the last few years. The authors in \cite{liu2017finite} consider the problem of designing a controller for a linear switched system under delay and external disturbance with finite- and fixed-time convergence. In \cite{li2013robust}, the authors design a hybrid observer and show finite-time convergence in the presence of unknown, constant bias.  In \cite{nersesov2007finite}, the authors study FTS of nonlinear impulsive dynamical systems, and present sufficient conditions
to guarantee FTS. The work in \cite{li2013robust,nersesov2007finite} considers discrete jumps in the system states in a continuously evolving system, i.e., they consider one model for the continuous dynamics, and one model for the discrete dynamics. In this work, we consider the general case with $N_f$ continuous flows and $N_g$ discrete jump dynamics, where $N_f,N_g$ can be any positive integers. The authors in \cite{li2019finite} present conditions in terms of a common Lyapunov function for FTS of hybrid systems. They require the value of the Lyapunov function to be decreasing during the continuous flow and non-increasing at the discrete jumps. In \cite{zhangfinite}, the authors consider a switched system whose subsystems possess a homogeneous Lyapunov function and are of constant switching intervals. In \cite{orlov2004finite}, the authors introduce the concept of \textit{locally} homogeneous system and show FTS of switched systems with uniformly bounded uncertainties. More recently, \cite{zhang2018finite} studies FTS of homogeneous switched systems by introducing the concept of hybrid homogeneous degree, and relating negative homogeneity with FTS. In \cite{fu2015global}, the authors consider systems in strict-feedback form with positive powers and design a controller as well as a switching law so that the closed-loop system is FTS. In this work, we do not assume that the subsystems of the hybrid system are homogeneous or in strict feedback form, and present conditions for FTS in terms of multiple generalized Lyapunov functions.

In this paper, we consider a general class of hybrid systems under arbitrary switching, and develop sufficient conditions for FTS of its equilibrium. We develop sufficient conditions for FTS of the origin of the hybrid system in terms of multiple generalized Lyapunov functions. We relax the requirement in \cite{li2019finite,li2016results} that the Lyapunov function is non-increasing at the discrete jumps and strictly decreasing during the continuous flow. More specifically, we use generalized Lyapunov functions for each mode, and allow these functions to increase both during the continuous flow and at the discrete jumps, and only require that these increments are bounded. In other words, we allow the hybrid system to have unstable modes. We show that if the origin is stable and if there exists an FTS mode that is active for a sufficiently long cumulative time, then the origin is FTS under arbitrary switching. In our earlier work \cite{garg2019finite}, we have presented conditions for FTS of switched systems. The current paper extends these results to a general class of hybrid systems. To the best of authors' knowledge, this is a first work on FTS of hybrid systems using multiple generalized Lyapunov functions. We summarize in Remark \ref{remark: summary contri} how earlier results on FTS of continuous-time,  switched and hybrid systems are special cases of the results presented in this paper.

The paper is organized as follows: In Section \ref{FTS HS}, we present an overview of FTS followed by conditions for FTS of hybrid systems under arbitrary switching. Section \ref{Simulations} evaluates the performance of the proposed method via simulation results. Our conclusions and thoughts on future work are summarized in Section \ref{Conclusions}.

\section{FTS of Hybrid Systems}\label{FTS HS}

\subsection{Notations}

We denote by $\|\cdot\|$ the Euclidean norm of vector $(\cdot)$, $|\cdot|$ the absolute value if $(\cdot)$ is scalar and the length if $(\cdot)$ is a time interval. The set of non-negative reals is denoted by $\mathbb R_+ = [0,\infty)$, set of non-negative integers by $\mathbb Z_+$ and set of positive integers by $\mathbb N$. We denote by $\textrm{int}(S)$ the interior of the set $S$, and by $t^-$ and $t^+$ the time just before and after the time instant $t$, respectively. 

\vspace{-10pt}
\subsection{Overview of FTS}

Consider the system: 
\begin{equation}\begin{split}\label{ex sys}
\dot y(t) = f(y(t)),
\end{split}\end{equation}
where $y\in \mathbb R^n$, $f: D \rightarrow \mathbb R^n$ is continuous on an open neighborhood $D\subseteq \mathbb R^n$ of the origin and $f(0)=0$. The origin is said to be an FTS equilibrium of \eqref{ex sys} if it is Lyapunov stable and \textit{finite-time convergent}, i.e., for all $y(0) \in \mathcal N \setminus\{0\}$, where $\mathcal N$ is some open neighborhood of the origin, $\lim_{t\to T} y(t)=0$, where $T = T(y(0))<\infty$  \cite{bhat2000finite}.
The authors also presented Lyapunov conditions for FTS of \eqref{ex sys}:

\begin{Theorem}[\cite{bhat2000finite}]\label{FTS Bhat}
Suppose there exists a continuous function $V$: $D \rightarrow \mathbb{R}$ such that the following hold: \\
(i) $V$ is positive definite \\
(ii) There exist real numbers $c>0$ and $\alpha \in (0, 1)$ , and an open neighborhood $\mathcal{V}\subseteq D$ of the origin such that 
\begin{equation}\begin{split} \label{FTS Lyap}
    \dot V(y) + cV(y)^\alpha \leq 0, \; y\in \mathcal{V}\setminus\{0\},
\end{split}\end{equation}
where $\dot V(y) \triangleq (D^+(V\circ\psi^x))(0)$ is the upper-right Dini derivative of $V(y)$. Then the origin is a FTS for \eqref{ex sys}.
\end{Theorem}

\subsection{Main result}

In this section, we consider the class of hybrid systems $\mathcal H = \{\mathcal F, \mathcal G , C, D\}$ described as
\begin{equation}\label{hybrid sys}
\begin{split}
    \dot x(t) & = f_{\sigma_f}(x(t)) \quad x(t)\in C,\\
    x(t^+) & = g_{\sigma_g}(x(t)) \quad x(t)\in D,
\end{split}\end{equation}
where $x\in \mathbb R^n$ is the state vector with $x(t_0) = x_0$, $f_{i}(x)\in \mathcal F \triangleq\{f_i\}$ for $i\in \Sigma_f \triangleq\{1,2,\cdots, N_f\}$ is the continuous flow for the system which are allowed on the subset of the state space $C\subset\mathbb R^n$ and $g_{j}(x) \in \mathcal G\triangleq\{g_j\}$ for $j\in \Sigma_g \triangleq\{1,2,\cdots, N_g\}$ defines the discrete behavior, which is allowed on the subset $D\subset\mathbb R^n$. Define $x^+(t) \triangleq x(t^+)$. The switching signals $\sigma_f$ and $\sigma_g$ are assumed to be piecewise constant, right-continuous signals, which can depend upon both state and time. Details about definitions and solution concepts of the hybrid system \eqref{hybrid sys} can be found in \cite{goebel2012hybrid}. 

Denote by $T_{i_k} =  [t_{i_k},t_{i_k+1})$ the interval in which the flow $f_i$ is active for the $k-$th time for $i\in \Sigma_f$ and $k\in \mathbb N$, and $t^d_{j_m}$ as the time when discrete jump $x^+ = g_j(x)$ takes place for the $m-$th time for $j\in \Sigma_g$ and $m\in \mathbb N$. Define $J_{i} = \{t^d_{j_m}\; | t^d_{j_m}\in T_{i_k}, j\in \Sigma_g, m\in \mathbb N\}$ as the set of all time instances when a discrete jump takes place when the continuous flow $f_i$ is active. We assume that the solution $x(t)$ to \eqref{hybrid sys} exists and is maximal. For each interval $T_{i_k}$, define the largest connected sub-interval $\bar T_{i_k}\subset T_{i_k}$, such that there is no discrete jump in $\bar T_{i_k}$, i.e., $\textrm{int}(\bar T_{i_k})\bigcap J_{i_k} = \emptyset$. For example, if $T_{i_1} = [0,1)$ and $J_{i} = \{0.2, 0.4, 0.75\}$, then $\bar T_{i_1} = [0.4, 0.75)$. We make the following assumptions:

\begin{Assumption}\label{assum uniq equib}
The origin is the only equilibrium point for all the continuous flows and discrete jumps, i.e., $f_i(0) = g_j(0) = 0$ for all $f_i\in \mathcal F$ and $g_j\in \mathcal G$.
\end{Assumption}

In this work, per Assumption \ref{assum uniq equib}, we restrict our focus on the case when there is a unique equilibrium of the hybrid system \eqref{hybrid sys}. The case when there exists some $g_j\in \mathcal G$ and a set $\bar D\neq\{0\}$ such that $g_j(x) = 0$ for all $x\in \bar D\subset D$, can be treated by studying the stability property of set $\bar D$; see \cite{li2019finite,li2016results}. 

\begin{Assumption}\label{assum bar T tm}
The time interval $\bar T_{F_k}$ for the FTS mode $F\in \Sigma$ is such that its length satisfies $|\bar T_{F_k}| \geq t_d>0$ for all $k\in \mathbb N$. 
\end{Assumption}

Assumption \ref{assum bar T tm} implies that for the FTS flow $f_F$, in each interval $T_{F_k}$ when the system \eqref{hybrid sys} evolves along the flow $f_F$, there exists a sub-interval $\bar T_{F_k}$ of non-zero length $t_d$ such that there is no discrete jump in the system state during $\bar T_{F_k}$. Let $\bar T_{F_k} = [\bar t_{F_k},\bar t_{F_k+1})$ with $ \bar t_{F_k+1}-\bar t_{F_k}\geq t_d$, and $\{\bar V_{F_1}, \bar V_{F_2}, \cdots, \bar V_{F_p}\}$ and $\{\bar V_{F_1+1}, \bar V_{F_2+1}, \cdots, \bar V_{F_p+1}\}$ be the sequence of the values of the generalized Lyapunov function $V_F$ at the beginning and end of the intervals $\bar T_{i_k}$, respectively. Let $i^0, i^1, \cdots,i^l,\cdots\in\Sigma_f$ be the sequence of modes active on the intervals $[t_{0}, t_{1}), [t_{1},t_{2}), \cdots,[t_l,t_{l+1}),\cdots$, respectively. The conditions for FTS of the origin of \eqref{hybrid sys} in terms of a common Lyapunov function are already given in \cite{li2019finite,li2016results}. In this work, we direct our focus on conditions in terms of multiple generalized Lyapunov functions defined as follows.

\begin{Definition}\label{gen Lyap func}
(\textbf{Generalized Lyapunov Function}): A continuous, positive definite function $V_j:\mathbb R^n\rightarrow\mathbb R_+$ is called a generalized Lyapunov function if there exists a continuous function $\phi:\mathbb R_+\rightarrow \mathbb R_+$ satisfying $\phi(0) = 0$, such that $V_j(x(t)) \leq \phi(V_j(x(t_{j_i})))$ for all $t\in [t_{j_i}, t_{j_i+1})$. 
\end{Definition}

\begin{Definition}\label{class GK}
(\textbf{Class-$\mathcal{GK}$ function}): A function $\alpha:\mathbb R_+\rightarrow\mathbb R_+$ is called a class-$\mathcal{GK}$ function if it is increasing, i.e., for all $x>y\geq 0$, $\alpha(x)>\alpha(y)$, and right continuous at the origin with $\alpha(0) = 0$.
\end{Definition}

\begin{Definition}\label{class GK infty}
(\textbf{Class-$\mathcal{GK}_\infty$ function}): A function $\alpha:\mathbb R_+\rightarrow\mathbb R_+$ is called a class-$\mathcal{GK}_\infty$ function if it is a class-$\mathcal{GK}$ function, and $\lim_{r\to\infty}\alpha(r) = \infty$. 
\end{Definition}

Note that the class-$\mathcal{GK}$ functions have similar composition properties as those of class-$\mathcal K$ functions, i.e., for $\alpha_1, \alpha_2\in \mathcal{GK}$ and $\alpha\in \mathcal{K}$, 
we have:
\begin{itemize}
    \item $\alpha_1\circ \alpha_2\in \mathcal{GK}$ and $\alpha_1+\alpha_2 \in \mathcal{GK}$;
    \item $\alpha\circ\alpha_1\in \mathcal{GK}$, $\alpha_1\circ\alpha\in \mathcal{GK}$ and $\alpha_1 + \alpha \in \mathcal{GK}$.
\end{itemize}


\begin{Lemma}\label{V diff lemma}
If a generalized Lyapunov function $V_F$ satisfies
\begin{equation}\begin{split}\label{v f equiv}
    \sum\limits_{k = 1}^p |V_F(x(t_{F_{k+1}})) - V_F(x(t_{F_k+1}))|\leq \alpha(\|x_0\|),
\end{split}\end{equation}
 for all $p \in \mathbb N$ where $\alpha \in \mathcal{GK}$, then there exists $\bar\alpha\in \mathcal{GK}$ such that the following holds for all $0\leq\beta<1$ and all $p\in \mathbb N$:
{\small
\begin{align}\label{v f actual}
\sum\limits_{k = 1}^p\Big(V_F(x(t_{F_{k+1}}))^{1-\beta} - V_F(x(t_{F_k+1}))^{1-\beta}\Big)\leq \bar\alpha(\|x_0\|).
\end{align}
}
\end{Lemma}
\begin{proof}
See Appendix \ref{app lemma 1 pf}.
\end{proof}

We now present our main result for FTS of hybrid systems.

\begin{Theorem}\label{FT Th 2}
The origin of \eqref{hybrid sys} is FTS if there exist generalized Lyapunov functions $V_i(x)$ for each $i\in \Sigma_f$ and the following hold:
\begin{itemize}
 \item[(i)] There exists $\alpha_1 \in \mathcal{GK}$, such that 
\begin{equation}
\begin{split}
   \sum\limits_{k = 0}^{p}\Big(
   V_{i^{k+1}}(x(t_{{k+1}})) -&V_{i^k}(x(t_{{k+1}}))\Big)\leq \alpha_1(\|x_0\|), \label{Hyb cond 1}
\end{split}\end{equation}
holds for all $p\in \mathbb Z_+$;
\item[(ii)] There exists $\alpha_2\in \mathcal{GK}$ such that
\begin{equation}\begin{split}
   \sum\limits_{k = 0}^{p}\Big(
   V_{i^k}(x(t_{{k+1}})) -&V_{i^k}(x(t_{k}))\Big)\leq \alpha_2(\|x_0\|), \label{Hyb cond 2}
\end{split}\end{equation}
holds for all $p\geq 0$;
\item[(iii)] There exists $\alpha_3\in \mathcal{GK}$ such that for all $i\in \Sigma_f$, 
\begin{equation}\begin{split}
    \sum\limits_{t\in J_{i}}\Big(V_i(x^+(t))-V_i(x(t))\Big)& \leq \alpha_3(\|x_0\|); \label{Hyb cond 3}
\end{split}\end{equation}
\item[(iv)] There exists $F\in \Sigma_f$ such that $\dot y = f_F(y)$ is FTS, and there exist a positive definite generalized Lyapunov function $V_F$ and constants $c>0$, $0<\beta<1$ such that 
\begin{equation}\label{v dot cond}
\dot V_F \leq -cV_F^\beta,    
\end{equation}
for all $t\in \bigcup [t_{F_k},t_{F_k+1})\setminus J_F$;
\item[(v)] There exists $\alpha_4\in \mathcal{GK}$ such that for all $p\geq 1$,
\begin{equation}\label{bar v f equiv}
\begin{split}
    \sum\limits_{k = 1}^p |\bar V_{F_{k+1}} - \bar V_{F_k+1}|\leq \alpha_4(\|x_0\|);
\end{split}\end{equation} 
\item[(vi)] The accumulated duration of activation for mode $F$ without any discrete state jump satisfy $\bar T_F = \sum_i|\bar T_{i_k}| = \gamma(\|x_0\|)$ where $\gamma\in \mathcal{GK}$ where $|\bar T_{i_k}|$ is the length of the interval $T_{i_k}$.
\end{itemize}
Moreover, if all the conditions hold globally, the functions $V_i$ are radially unbounded for all $i\in \Sigma_f$, and $\alpha_l\in \mathcal{GK}_\infty$ for $l\in \{1,2,3,4\}$, then the origin of \eqref{hybrid sys} is globally FTS (GFTS). 
\end{Theorem}


\begin{proof}
Let $x_0\in D$, where $D$ is some open neighborhood of the origin. For all $p\in \mathbb Z_+$, we have that
\begin{align*}
    V_{i^p}(x(t_p)) & = V_{i^0}(x(t_{0})) + \sum\limits_{k = 1}^{p}
   \Big(V_{i^k}(x(t_{k})) -V_{i^{k-1}}(x(t_{k}))\Big)\\
    & \quad + \sum\limits_{k = 0}^{p-1}
   \Big(V_{i^k}(x(t_{{k+1}})) -V_{i^k}(x(t_k))\Big)\\
   & \quad + \sum\limits_{k = 0}^{p}\sum_{t\in J_{k}\bigcap [t_k, t_{k+1})}
   \Big(V_{i^k}(x(t^+)) -V_{i^k}(x(t))\Big)\\
    & \overset{\eqref{Hyb cond 1},\eqref{Hyb cond 2},\eqref{Hyb cond 3}}{\leq} \alpha_0(\|x_0\|)+\alpha_1(\|x_0\|)\\
    & \quad \quad \quad +\alpha_2(\|x_0\|) + N\alpha_3(\|x_0\|),
\end{align*}
where $\alpha_0(r) = \max\limits_{i\in \Sigma, \|x\|\leq r}V_i(x)$. Since $V_i$ are positive definite, $\alpha_0$ is positive definite. Hence, there exists $\alpha\in \mathcal{GK}$ such that $\alpha_0(\|x_0\|)+\alpha_1(\|x_0\|)+\alpha_2(\|x_0\|) + N\alpha_3(\|x_0\|)\leq \alpha (\|x_0\|)$ for all $x_0\in D$. Consequently: 
\begin{align}\label{v f a}
 V_F(x(t_{F_i})) \leq \alpha(\|x_0\|).  
\end{align}

We also know that during $\bar T_{F_k}$, there is no discrete jump for all $k\in \mathbb N$. As per condition (iv), $\dot V_F \leq -cV_F^\beta$ for all $t\in \bigcup\bar T_{F_k}$. Using this and the Comparison lemma, we obtain:
\begin{equation*}\begin{split}
& |\bar T_{F_k}| \leq \frac{\bar V_{F_k}^{1-\beta}}{c((1-\beta)}-\frac{\bar V_{F_k+1}^{1-\beta}}{c(1-\beta)}\\
\implies \sum\limits_{k = 1}^M &|\bar T_{F_k}| \leq \sum\limits_{k = 1}^M\Big(\frac{\bar V_{F_k}^{1-\beta}}{c((1-\beta)}-\frac{\bar V_{F_k+1}^{1-\beta}}{c(1-\beta)}\Big)\\
& = \frac{\bar V_{F_1}^{1-\beta}}{c(1-\beta)} + \sum\limits_{i = 1}^
{M-1}\frac{\bar V_{F_{i+1}}^{1-\beta}-\bar V_{F_i+1}^{1-\beta}}{c(1-\beta)} -\frac{\bar V_{F_M+1}^{1-\beta}}{c(1-\beta)}. 
\end{split}\end{equation*}
Using \eqref{v f a}, we obtain that $\frac{V_{F_1}^{1-\beta}}{c(1-\beta)} \leq \frac{\alpha(\|x_0\|)^{1-\beta}}{c(1-\beta)}$. Define $\gamma_1(\|x_0\|) \triangleq \frac{\alpha(\|x_0\|)^{1-\beta}}{c(1-\beta)}$ so that $\gamma_1\in \mathcal{GK}$. From condition (v) and Lemma \ref{V diff lemma}, there exists $\bar \alpha\in \mathcal{GK}$ such that: 
\begin{align}\label{v f need}
\sum\limits_{i = 1}^{M-1}\frac{\bar V_{F_{i+1}}^{1-\beta}-\bar V_{F_i+1}^{1-\beta}}{c(1-\beta)} \leq \frac{\bar \alpha(\|x_0\|)}{c(1-\beta)}.
\end{align}
Define $\gamma(\|x_0\|) \triangleq \gamma_1(\|x_0\|) + \frac{\bar \alpha(\|x_0\|)}{c(1-\beta)}$ so that we obtain:
\begin{equation*}\begin{split}
   \bar T_F + \frac{\bar V_{F_M+1}^{1-\beta}}{c(1-\beta)} &\leq \frac{\bar V_{F_1}^{1-\beta}}{c(1-\beta)} + \sum\limits_{i = 1}^
{M-1}\frac{\bar V_{F_{i+1}}^{1-\beta}-\bar V_{F_i+1}^{1-\beta}}{c(1-\beta)}\\
& \leq \gamma(\|x_0\|).
\end{split}\end{equation*}
Clearly, $\gamma \in \mathcal{GK}$. Now, with $\bar T_F  = \gamma(\|x_0\|)$, we obtain 
\begin{equation*}\begin{split}
     \bar T_F+\frac{\bar V_{F_M+1}^{1-\beta}}{c(1-\beta)}\leq \gamma(\|x_0\|) =  \bar T_F,
\end{split}\end{equation*}
which implies $\frac{\bar V_{F_M+1}^{1-\beta}}{c(1-\beta)} \leq 0$. But $\bar V_F\geq 0$, which implies $\bar V_{F_M+1} = 0$. Hence, if mode $F$ is active for the accumulated time $\bar T_F$ without any discrete jump in the system state, the value of the function $V_F$ converges to $0$ as $t\rightarrow \bar t_{F_M+1}$. 

From Assumption \ref{assum bar T tm}, $|\bar T_{F_k}|\geq t_d$ for all $k\in \mathbb N$, and hence $Mt_d \leq \sum\limits_{i = 1}^M|\bar T_{F_k}| =  \gamma(\|x_0\|)$, which implies that $M\leq \frac{\gamma(\|x_0\|)}{t_d} <\infty$. Now we show that the time of convergence is finite, i.e., $\bar t_{F_M+1} <\infty$. From the above analysis, we have that if $\sum\limits_{i = 1}^M |\bar T_{F_i}| =  \bar T_F$, then there exists an interval $[\bar t_{F_M},\bar t_{F_M+1})$ such that $\bar V_{F_M+1} = 0$. Since $\bar V_{F_i} \leq \bar\alpha(\|x_0\|)<\infty$, we obtain that
\begin{align*}
  \bar t_{F_M+1}-\bar t_{F_M}& \leq \frac{\bar V_{F_M}^{1-\beta}-\bar V_{F_M+1}^{1-\beta}}{c(1-\beta)}\leq \frac{\bar V_{F_M}^{1-\beta}}{c(1-\beta)} <\infty  
\end{align*}
for all $x_0\in D$. Now, there are two cases possible. If $\bar t_{F_M}<\infty$, then, we obtain that $\bar t_{F_M+1} \leq \bar t_{F_M} + \frac{\bar\alpha(\|x_0\|)^{1-\beta}}{c(1-\beta)}<\infty$ for all $x_0\in D$. If $\bar t_{F_M} = \infty$, we obtain that the time of activation for mode $F$ reads $\bar T_F = \sum\limits_{i = 1}^{M-1}|\bar T_i| <\gamma(\|x_0\|)$, which contradicts condition (vi). Hence, for condition (vi) to hold, it is required that $\bar t_{F_M}<\infty$ and hence, we obtain that $\bar t_{F_M+1}<\infty$. Hence, the trajectories of \eqref{hybrid sys} reach the origin with finite number of active intervals of the continuous flow $f_F$. 

If all the conditions (i)-(vi) hold globally and the functions $V_i$ are radially unbounded, we obtain that $\alpha_0$ is also radially unbounded. Using the fact that $\alpha_1, \alpha_2, \alpha_3,\alpha_4\in \mathcal{GK}_\infty$, we have  $\bar\alpha(\|x_0\|) <\infty$ for all $\|x_0\|<\infty$, and hence, we obtain $t_{F_M+1} \leq t_{F_M} + \frac{\bar\alpha(\|x_0\|)^{1-\beta}}{c(1-\beta)}<\infty$ for all $x_0$, which implies GFTS.  
\end{proof}

\begin{figure}[!htbp]
	\centering
	\includegraphics[width=0.92\columnwidth,clip]{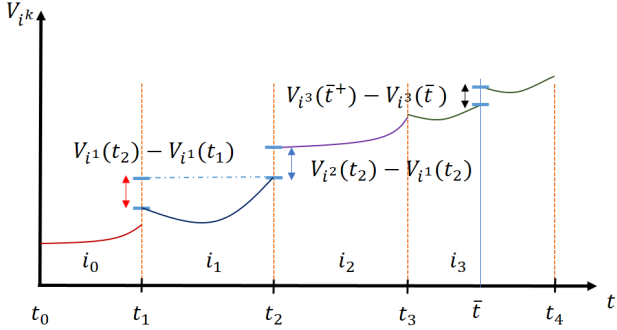}
	\caption{Conditions (i), (ii) and (iii) of Theorem \ref{FT Th 2}: change in the values of the generalized Lyapunov functions. The increment shown by blue, red and black double-arrows pertain to condition (i), (ii) and (iii), respectively. }
	\label{fig:V i incr}
\end{figure}

\begin{figure}[!htbp]
	\centering
	\includegraphics[width=0.92\columnwidth,clip]{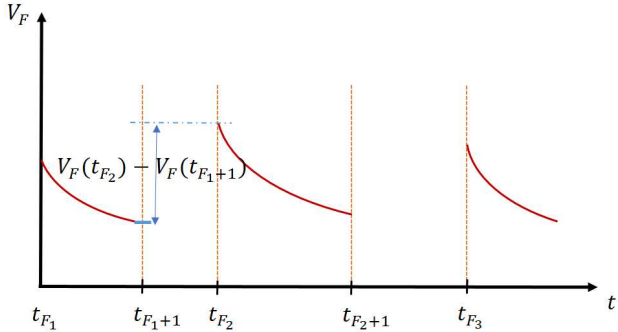}
	\caption{Increment in the value of $V_F$ during the switched-off period (condition (v) of Theorem \ref{FT Th 2}).}
	\label{fig:V F incr}
\end{figure}

\noindent Condition (i) means that at switching instants, the cumulative value of the differences between the consecutive Lyapunov functions is bounded by a class-$\mathcal{GK}$ function; condition (ii) means that the cumulative increment in the value of the individual Lyapunov functions when the respective modes are active, is bounded by a class-$\mathcal{GK}$ function (see Figure \ref{fig:V i incr}). Some authors use the time derivative condition, i.e., $\dot V_i\leq \lambda V_i$ with $\lambda>0$, in place of condition (ii) to allow growth of $V_i$, hence, requiring the function to be continuously differentiable (e.g. \textnormal{\cite{wang2018conditions}}). Our condition allows the use of non-differentiable generalized Lyapunov function; condition (iii) means that the cumulative increment in the value of the generalized Lyapunov function $V_i$ is bounded by a class-$\mathcal{GK}$ function during the discrete jumps; condition (iv) means that there exists an FTS mode $F\in \Sigma$; condition (v) means that the accumulated increment in the value of the Lyapunov function $V_F$ during the "switched-off" periods is bounded (see Figure \ref{fig:V F incr}); and condition (vi) means that the FTS mode $F$ is active for a minimum accumulated time.

\begin{Remark}
In general, the inequality \eqref{v f need} can be difficult to obtain directly. Consider the case when we only know that the mode $F$ is homogeneous, with negative degree of homogeneity. From \cite[Theorem 7.2]{bhat2005geometric}, we know that in this case, condition (iii) of Theorem \ref{FT Th 2} holds for some $\beta$, but its exact value might not be known. In this case, it is not possible to bound the left-hand side (LHS) of \eqref{v f need}. Lemma \ref{V diff lemma} allows one to bound this LHS with a class-$\mathcal{GK}$ function without explicitly knowing the value of $\beta$. \end{Remark}

Compared to \cite{li2019finite,li2016results}, our method is less conservative since we allow generalized Lyapunov functions to increase during the continuous flow (per \eqref{Hyb cond 2}) as well as at the discrete jump (per \eqref{Hyb cond 3}). Also, during the continuous flow the generalized Lyapunov functions are allowed to grow when switching from one continuous flow to another (per \eqref{Hyb cond 1}), whereas the aforementioned work imposes that the common Lyapunov function is always non-increasing. 

We discuss the relation of Assumption \ref{assum bar T tm} with the average dwell-time (ADT) for discrete jumps method that is often used in the literature \cite{teel2013lyapunov}. Discrete jumps under ADT means that in any given interval $[t_1, t_2]$, the number of discrete jumps of the states $N([t_1,t_2])$ satisfies $N([t_1,t_2])\leq N_0 + \delta(t_2-t_1)$, where $\delta\geq 0$ and $N_0\geq 1$. (see \cite{teel2013lyapunov,cai2008smooth}). We note that if the conditions of the ADT method hold for the mode $F$, instead of the minimum dwell-time condition of Assumption \ref{assum bar T tm}, then the FTS result still holds since the parameter $M$ used in the proof of Theorem \textnormal{\ref{FT Th 2}} can be defined as $M = N_0 + \frac{T_F}{\tau_a}$, where $T_F$ is the total time of activation of mode $F$.

So far we have considered the cases when only one of the modes of the system \eqref{hybrid sys} is FTS. Next, we present a result for the case when all the continuous modes of the hybrid system are FTS. 

\begin{Theorem}
The origin of \eqref{hybrid sys} is FTS if the origin of $\dot x = f_i(x)$ is FTS for all $i\in \Sigma_f$ and there exist generalized Lyapunov functions $V_{i}$ for each $i\in \Sigma_f$ such that conditions (i), (iii) and (v) of Theorem \ref{FT Th 2} are satisfied.
\end{Theorem}

\begin{proof}
Let $T_i$ denote the total time duration for which mode $i\in \Sigma$ is active without any discrete jump in the state. As the generalized 
Lyapunov function $V_{i}$ satisfies condition (iv) of Theorem \eqref{FT Th 2}, it satisfies the condition (ii) also with $\alpha_2 = 0$. Let $T_F(\|x_0\|) = \max\limits_i \gamma_i(\|x_0\|)$. In any time interval $[0, T)$, at least one mode $i$ satisfies $T_i \geq \frac{T}{N}$. Hence, for $T\geq T_M \triangleq NT_F(\|x_0\|)$, we obtain that there exists at least one mode $i$ 
satisfying $T_i \geq \frac{T}{N} \geq \frac{T_M}{N} = T_F(\|x_0\|)\geq \gamma_i(\|x_0\|)$. This implies that all the conditions of Theorem \ref{FT Th 2} are satisfied. Therefore, we obtain that the origin of \eqref{hybrid sys} is FTS.
\end{proof}

\begin{Remark}\label{remark: summary contri}
For switched systems, i.e., systems with no discrete jump, Theorem \ref{FT Th 2} guarantees FTS of the origin with conditions (i), (ii), (iv), (v) and (vi), since with $D = \emptyset$, condition (iii) is not needed. Hence, the results presented in this paper are applicable to switched systems as well. Furthermore, if $N_f = 1$, i.e., if the system \eqref{hybrid sys} reduces to a continuous-time dynamical system, Theorem \ref{FT Th 2} reduces to Theorem \ref{FTS Bhat}. Thus, the continuous-time result for FTS is a special case of Theorem \ref{FT Th 2}. Finally, some of the results presented in \cite{li2019finite} are special cases of Theorem \ref{FT Th 2} with $V = V_F$ and $\alpha_3 = 0$.
\end{Remark}

\vspace{-10pt}

\section{Simulations}\label{Simulations}
We present a numerical example with number of modes $N = 5$, where the fifth mode is FTS, i.e., $F = 5$. The  system modes are chosen as 
\begin{align*}
    &f_1 = \begin{bmatrix}0.01x_1^2+x_2\\-0.01x_1^3+x_2\end{bmatrix},\quad     f_2 = \begin{bmatrix}0.01x_1-x_2\\-x_1^2+0.01x_2\end{bmatrix},\\
    &f_3 = \begin{bmatrix}-x_1-x_2\\x_1-x_2\end{bmatrix},\quad
    f_4 = \begin{bmatrix}0.01x_1^2+0.01x_1x_2\\-0.01x_1^3+x_2^2\end{bmatrix},\\
    &\qquad \quad f_5 = \begin{bmatrix} x_2-20\textrm{sign}(x_1)|x_1|^\alpha\\-10\textrm{sign}(x_1)|x_1|^{2-2\alpha}\end{bmatrix}.
\end{align*}
The generalized Lyapunov functions are chosen as $V_i(x) = x^TP_ix$ where matrices $P_i$ are chosen as $P_1 = \begin{bmatrix}1 & 0\\ 0 & 1\end{bmatrix},  P_2 = \begin{bmatrix}5 & 2\\ 2 & 4\end{bmatrix}, P_3 = \begin{bmatrix}1 & 0\\ 0 & 3\end{bmatrix}, P_4 = \begin{bmatrix}6 & 1\\ 2 & 3\end{bmatrix}$, and $V_5(x) = \frac{k_2}{2\alpha}|x_1|^{2\alpha} + \frac{1}{2}|x_2|^2$. The discrete jumps are defined by $g = \begin{bmatrix}-1.1x_1 \\-1.1x_2\end{bmatrix}$ so that the states $x_1$ and $x_2$ change their signs and increase in the magnitude at the discrete jumps. Note that this example is more general than the examples considered in \cite{li2019finite}, as we allow the dynamics to have unstable modes. In this case, we allowed the switches in the continuous flows after $0.2$ sec, i.e., $|T_{i_k}| = 0.2$ sec and discrete jumps after $0.1$ sec, so that $t_d = 0.1$ sec, i.e., $|\bar T_{i_k}| = 0.1$, for all $i\in \{1, 2,\cdots, 5\}$ and $k\in \mathbb Z_+$ (see Assumption \ref{assum bar T tm}). Figure \ref{fig:sim hyb mode} shows the switching signal $\sigma_f(t)$.  Figure \ref{fig:sim hyb states} shows the states $x_1(t)$ and $x_2(t)$ for the first 10 seconds of the simulation. The states change sign at the discrete jumps. It can be seen in the figure that the system moves away from the origin while operating in the unstable modes. Figure \ref{fig:sim hyb norm} shows the norm of the state vector $x(t)$ on log scale. It can be seen that the system states reach the origin in finite time. Figure \ref{fig:sim hyb V} shows the evolution of generalized Lyapunov functions $V_i$ with respect to time for first 10 seconds of the simulation. As can be seen from the figure, the generalized Lyapunov functions increase at the time of discrete jumps, the switching instants, as well as during the continuous flows along the unstable modes. 

\begin{figure}[!htbp]
	\centering
	\includegraphics[width=0.92\columnwidth,clip]{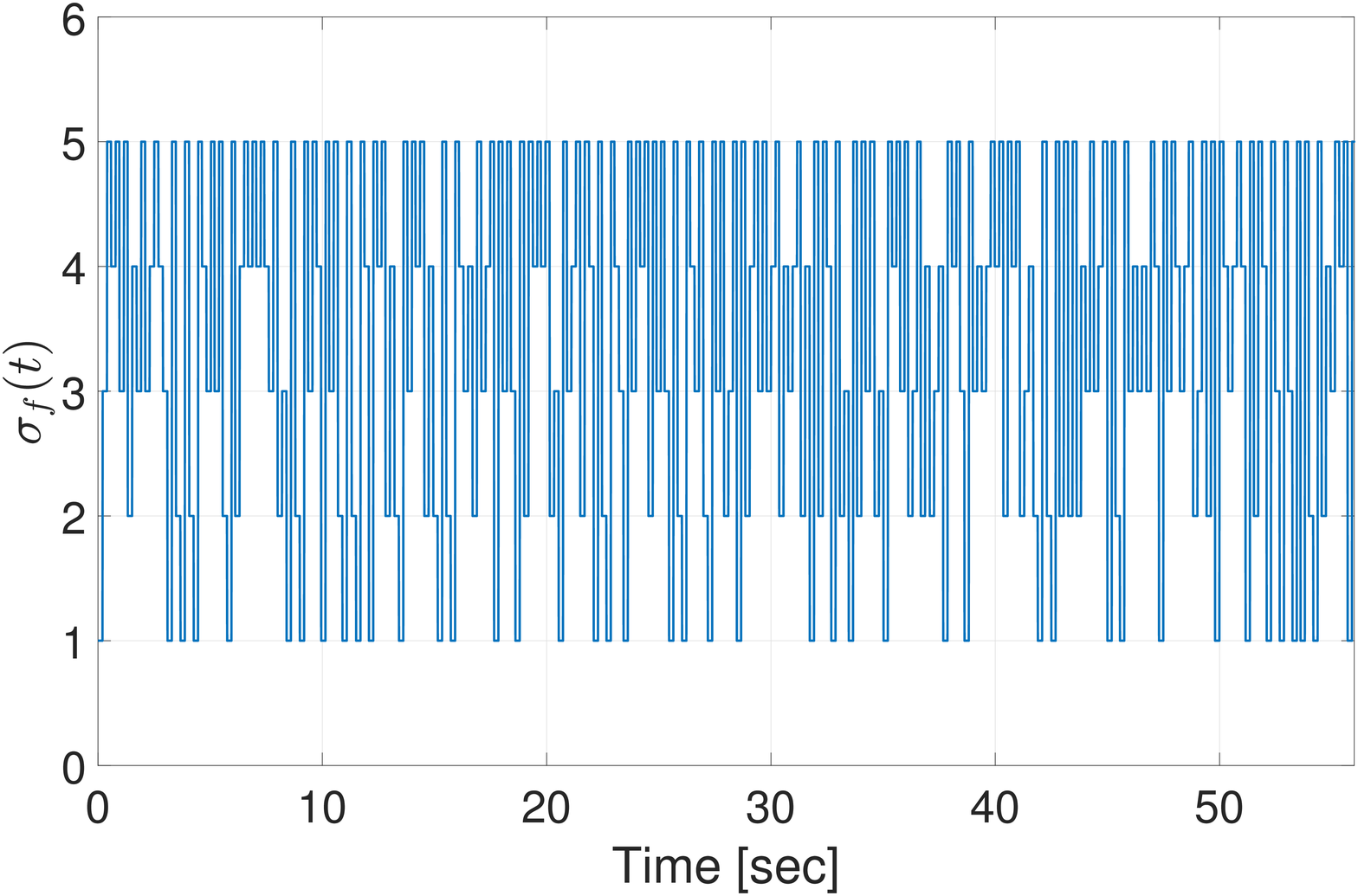}
	\caption{Switching signal $\sigma_f(t)$ for the considered hybrid system.}
	\label{fig:sim hyb mode}

	\centering
	\includegraphics[width=0.92\columnwidth,clip]{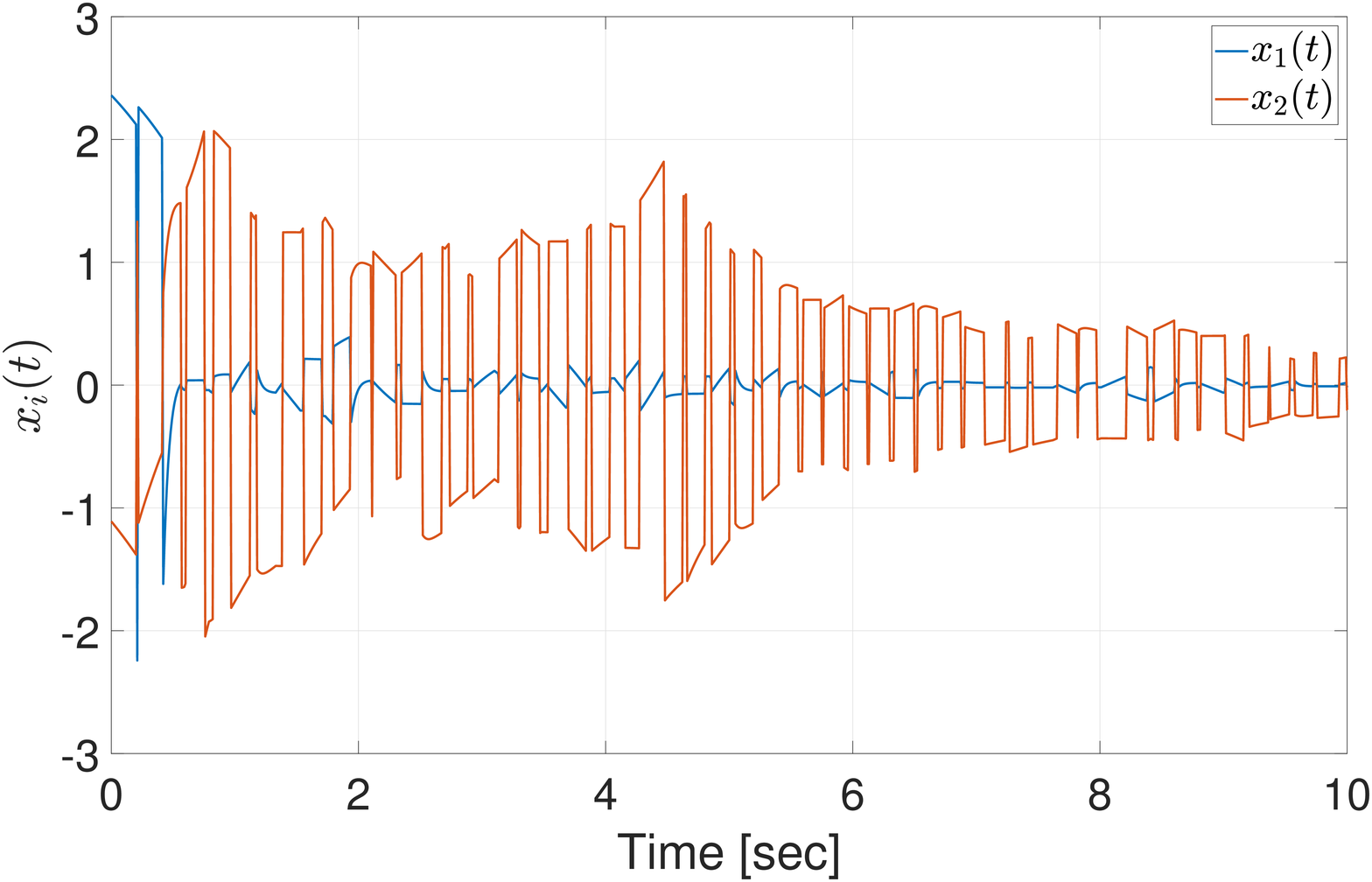}
	\caption{System states $x_1(t)$ and $x_2(t)$ for $t\in [0,10]$ sec. }
	\label{fig:sim hyb states}

	\centering
	\includegraphics[width=0.92\columnwidth,clip]{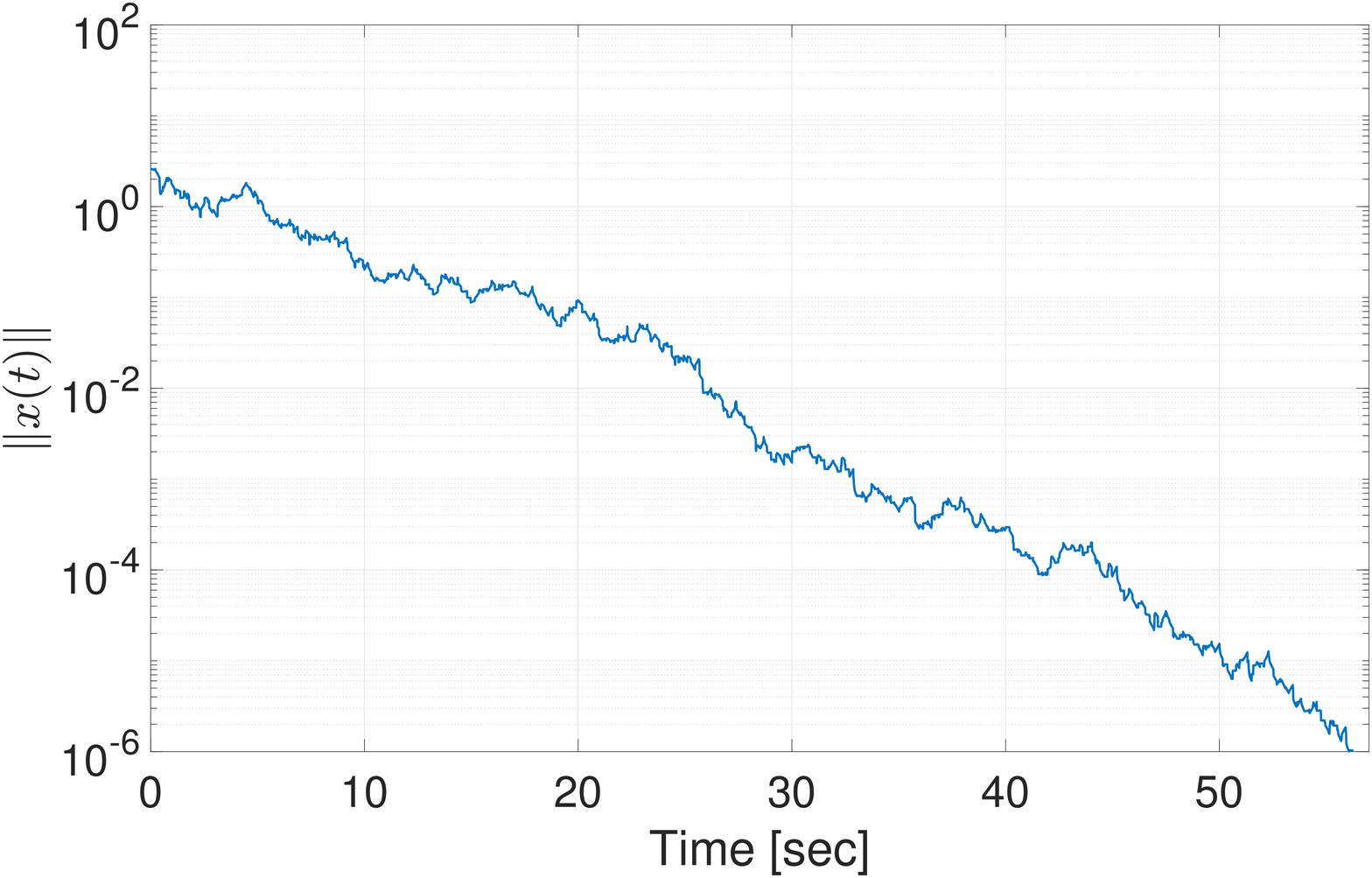}
	\caption{The norm of the state vector $x(t)$ for the considered hybrid system.}
	\label{fig:sim hyb norm}
\end{figure}

\begin{figure}[!htbp]
	\centering
	\includegraphics[width=0.92\columnwidth,clip]{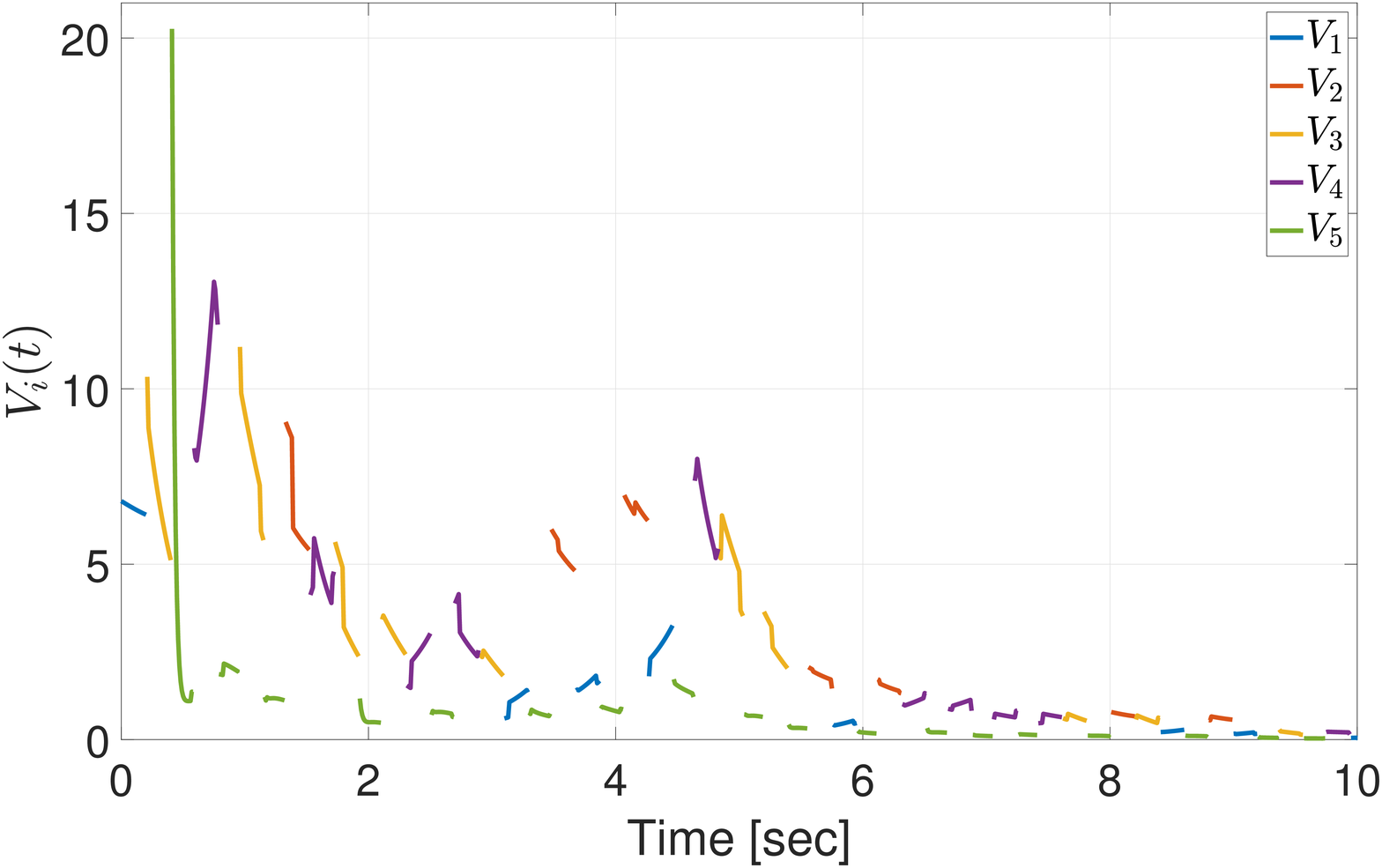}
	\caption{Generalized Lyapunov functions $V_i(t)$ for $t\in [0, 10]$ sec for the considered hybrid system.}
	\label{fig:sim hyb V}
\end{figure}

The provided example validates that the system can achieve FTS even if one or more modes are unstable, if the FTS mode is active for long enough. 

\section{Conclusions and Future Work}\label{Conclusions}
In this paper, we studied FTS of hybrid systems. We showed that under some mild conditions on the bounds on the difference of the values of generalized Lyapunov functions, if the FTS mode is active for a minimum required time, then FTS can be achieved. Our proposed method allows the individual generalized Lyapunov functions to increase both during the continuous flow as well as at the discrete jumps, i.e., it allows for the hybrid system to have unstable modes. 

Traditionally, optimization-based formulations are used to find control inputs that satisfy the Lyapunov conditions. Our ongoing research focuses on incorporating input and state constraints in the hybrid framework for specifications involving both spatial and temporal requirements. More specifically, we are investigating how to incorporate prescribed-time, rather than merely finite-time, convergence for the system modes, so that the overall framework can be used in the synthesis and analysis of controllers with spatiotemporal specifications. 


\vspace{-10pt}
\appendices

\section{Proof of Lemma 1}
\begin{proof}\label{app lemma 1 pf}
Let the function $V_F$ satisfy \eqref{v f equiv} for all $k\mathbb N$, where $\alpha\in \mathcal{GK}$. Lemma 3.3, 3.4 of \cite{zuo2016distributed} establish the following set of inequalities for $z_i\geq 0$ and $0<r\leq 1$
\begin{align}
    \left(\sum_{i = 1}^M z_i \right)^r \leq \sum_{i = 1}^M z_i^r \leq M^{1-r} \left(\sum_{i = 1}^M z_i \right)^r.
\end{align}
Hence, we have that for $a\geq b\geq 0$ and $0<r\leq 1$, $ a^r = (b + (a-b))^r \leq b^r  + (a-b)^r$, or equivalently, 
\begin{equation}\begin{split}
    a^r-b^r \leq (a-b)^r.
\end{split}\end{equation}
Define $a_i \triangleq V_F(x(t_{F_{i+1}}))$ and $b_i \triangleq V_F(x(t_{F_i+1}))$ so that whenever $a_i\geq b_i$, or, $V_F(x(t_{F_{i+1}})) \geq V_F(x(t_{F_i+1}))$, we have that $a_i^r-b_i^r \leq (a_i-b_i)^r$. Denote $i \in I_1 \triangleq \{i_1,i_2,\cdots, i_m\}$ for some $m\leq k$, for which $a_i\geq b_i$ and $i\in I_2\triangleq \{i_{m+1}, i_{m+2},\cdots, i_k\}$ for which $a_i\leq b_i$ (or equivalently, $a_i^r-b_i^r\leq 0$). Hence, we have that for any $0<r\leq 1$, 
\begin{equation*}\begin{split}
    \sum_{i = 1}^k (a_i^r-b_i^r) & \leq \sum_{i\in I_1}(a_i^r-b_i^r) \leq \sum_{i\in I_1}(a_i-b_i)^r\\
    & \leq  m^{1-r}\left(\sum_{i\in I_1}(|a_i-b_i|)\right)^r.
\end{split}\end{equation*}

Using this, we obtain that 

\begin{equation*}\begin{split}
  & \sum_{i = 1}^k ( V_F(x(t_{F_{i+1}}))^r-V_F(x(t_{F_i+1}))^r)\\
  \leq & m^{1-r}\left(\sum_{i \in I_1} (| V_F(x(t_{F_{i+1}}))-V_F(x(t_{F_i+1}))|)\right)^r\\
  \leq & m^{1-r}\left(\sum_{i = 1}^k | V_F(x(t_{F_{i+1}}))-V_F(x(t_{F_i+1}))|\right)^r\\
  \leq & m^{1-r}\alpha(\|x(0)\|)^r \triangleq \bar \alpha(\|x(0)\|)
\end{split}\end{equation*}

where $\bar \alpha$ is clearly a class $\mathcal{GK}$ function. Hence, we obtain that \eqref{v f actual} holds for all $0<\beta = 1-r<1$.
\end{proof}

\bibliographystyle{IEEEtran}
\bibliography{myreferences}

\end{document}